\documentclass{article}
\usepackage{rompmod}

\usepackage{dcolumn}
\usepackage{bm}

\usepackage{amsthm}
\usepackage{amssymb}

\usepackage{amstext}
\usepackage{amsmath}
\usepackage{amsfonts}
\usepackage{units}
\usepackage{mathrsfs}
\usepackage{cite}
\newcommand{\p}{\partial}

\newcommand{\dd}{{\rm d}}

\theoremstyle{theorem}
\newtheorem{theorem}{Theorem}

\theoremstyle{theorem}
\newtheorem{proposition}{Proposition}

\theoremstyle{theorem}
\newtheorem{corollary}{Corollary}

\theoremstyle{lemma}
\newtheorem{lemma}{Lemma}
\newtheorem{remark}[section]{Remark}

\theoremstyle{definition}

\theoremstyle{example}
\newtheorem{example}{Example}

\begin{document}

\title{A divergence theorem for pseudo-Finsler spaces}

\author{E. Minguzzi\thanks{
Dipartimento di Matematica e Informatica ``U. Dini'', Universit\`a
degli Studi di Firenze, Via S. Marta 3,  I-50139 Firenze, Italy.
 e-mail: ettore.minguzzi@unifi.it }}


\maketitle


\begin{abstract}
\noindent We study the divergence theorem on pseudo-Finsler spaces and obtain a completely Finslerian version for spaces having a vanishing mean Cartan torsion. This result  helps to clarify the problem of energy-momentum conservation in Finsler gravity theories.
\end{abstract}


\section{Introduction}
Two Finslerian divergence theorems  have been discussed  by H.\ Rund \cite{rund75} and Z.\ Shen \cite{shen98,shen01}. While they both equate a certain volume integral on a domain of $M$ with a surface integral at the boundary, they have quite a different nature.

Let $\pi\colon TM\to M$ be the tangent bundle, let $E=TM\backslash 0$, be the slit tangent bundle, and let $VE$ be the kernel of  $\pi_*$, namely the vertical tangent bundle. Let $\mathscr{L}\colon TM\backslash 0 \to \mathbb{R}$ be the Finsler Lagrangian (for detailed definitions we refer the reader to Sec.\ \ref{kqp}), and let $g$ be its vertical Hessian, namely the Finsler metric. Rund's theorem involves a vector field $Z\colon E\to TM$, a section $s\colon M\to E$, and the Finslerian divergence of $Z$ on $E$ calculated with the horizontal covariant derivative and pulled back to $M$ using $s$ (cf.\ Eq.\ (\ref{oap})). It has the drawback that the boundary term is not genuinely Finslerian, indeed, both the normal to the hypersurface and the boundary form are deduced from the pullback metric $s^*g$. So there appear elements of Riemannian geometry.

The version by Shen is somewhat complementary \cite[Theor.\ 2.4.2]{shen01}. It is less Finslerian for what concerns the vector field since it deals with  a field $X\colon M\to TM$ on the base, for the computation of whose divergence the  Finsler connection is not required. However,
the boundary term is genuinely Finslerian as it involves the notion of Finsler normal to a hypersurface.

In this work we are going to elaborate a further version of  the divergence theorem in Finsler geometry. We shall first give a short derivation of Rund's result and then we shall show that for pseudo-Finsler spaces with vanishing mean Cartan torsion, $I_\alpha=0$, it is possible to give a genuinely Finslerian divergence theorem in which both the vector field and the boundary terms are Finslerian.

It must be recalled that by Deicke's theorem \cite{deicke53,brickell65,bao00}  a Finsler space with vanishing mean Cartan torsion is actually Riemannian. As a consequence, the mentioned result will be of interest only for  pseudo-Finsler spaces of non-definite signature. It has been suggested by the author that Lorentz-Finsler spaces with zero mean Cartan torsion might be the appropriate objects of study in Finsler gravity theory \cite{minguzzi14c,minguzzi15c,minguzzi16c}, particularly because they have affine sphere indicatrices and because, as in general relativity, they are uniquely determined by a spacetime volume form and a light cone distribution  \cite{minguzzi15e}. Therefore, it is  expected that the results of this work could shed light on the problem of energy-momentum conservation in Finslerian extensions of general relativity.


\section{Connections in pseudo-Finsler geometry} \label{kqp}

In this section we assume some familiarity with the notion of Finsler connection and of pullback connection \cite{matsumoto86,antonelli93,ingarden93,anastasiei96,bejancu00,szilasi14,minguzzi14c}. Let us give some key coordinate expressions  in order to fix terminology and notations. Let $\{x^\mu\}$ be coordinates on a chart of $M$ and let $\{x^\mu, y^\mu\}$ be the induced coordinates on $TM$. The Finsler Lagrangian is, by definition, positive homogeneous of degree two $\mathscr{L}(x,sy)=s^2 \mathscr{L}(x,y)$, $\forall s>0$. Although we assumed that $\mathscr{L}$ is defined on the slit tangent bundle, this assumption can be relaxed, e.g.\ it could be defined on just a convex cone subbundle provided the next equations are evaluated on its domain, see \cite{minguzzi13c,minguzzi14h} for a complete discussion. The Finsler metric is given by the Hessian $g_{\mu \nu}=\frac{\p^2\mathscr{L}}{\p y^\mu \p y^\nu}$ and is assumed non-degenerate. The Cartan torsion is $C_{\alpha \beta \gamma}=\frac{1}{2}\frac{\p}{\p y^\alpha} g_{\beta \gamma}$ while the mean Cartan torsion is $I_\gamma=g^{\alpha \beta} C_{\alpha \beta \gamma}$.

\begin{example}
A non-trivial example of affine sphere spacetime is
\begin{align}
\mathscr{L}&=-\tfrac{2}{3^{3/4}}\left(\left(\tfrac{1}{2}  \dd t+\tfrac{\sqrt{3}}{2} a(t) \dd z\right)^{2}\right)^{1/4} \left(\!\left(\tfrac{\sqrt{3}}{2}  \dd t-\tfrac{1}{2} a(t) \dd z\right)^{\!\!2}- a^2(t) (\dd x^2+\dd y^2)\right)^{\!3/4} \nonumber
\end{align}
which in the low velocity limit (with respect to an observer whose velocity field is $u\propto\p_t$) gives the general relativistic Friedmann metric in the flat space section case (i.e.\ $k=0$) $g=-\dd t^2+ a(t)^2 (\dd x^2+\dd y^2+\dd z^2)$. Here $a(t)$ is the scale factor of the Universe. Many other examples of affine sphere spacetimes and a discussion can be found in \cite{minguzzi16c}. There are really plenty of affine sphere spacetimes since there are plenty of cone structures and volume forms. If the light cones have ellipsoidal sections one recovers the usual Lorentzian spacetimes \cite{minguzzi15e}. We remark that in what follows we do not assume a Lorentzian signature for the Finsler metric, though this is certainly the most interesting case.
\end{example}

 A non-linear connection is a splitting of the tangent space $TE=VE\oplus HE$ into vertical and horizontal bundles. A basis for the horizontal space is given by
\[
\Big\{\frac{\delta \ }{\delta x^\mu}\Big\}, \qquad \frac{\delta}{\delta x^\mu}=\frac{\p}{\p x^\mu}-N^\nu_\mu(x,y) \frac{\p}{\p y^\nu},
\]
where the coefficients $N^\nu_\mu(x,y)$ define the non-linear connection and have suitable transformation properties under change of coordinates.
The covariant derivative for the non-linear connection is defined as follows. Given a section $s\colon U\to E$, $U\subset M$,
\[
 D_{\xi} s^\alpha=\Big(\frac{\p s^\alpha}{\p x^\mu} +N^\alpha_\mu\big(x, s(x)\big) \Big)\xi^\mu .
\]
Usually one considers the non-linear connection determined by a spray as follows \[N^\mu_\alpha=G^\mu_\alpha:=\p G^\mu/\p y^\alpha ,\]
where
\[
2 {G}^\alpha(x,y)= g^{\alpha\delta}\Big( \frac{\p^2\mathscr{L} }{\p x^\gamma \p y^\delta} \,y^\gamma -\frac{\p\mathscr{L} }{\p x^\delta } \Big) .
\]
This will be our choice also.

The Finsler connections are splittings of the vertical bundle $\pi_E\colon VE\to E$.  In coordinates they are triplets of coefficients $(N^\alpha_\beta, H^\alpha_{\beta \gamma}, V^\alpha_{\beta \gamma})$ with suitable transformations properties. We shall assume that they are regular, namely $H^\alpha_{\beta \gamma} y^\beta=N^\alpha_\gamma$, and symmetric in the lower indices (torsionless).
The  Berwald, Cartan, Chern-Rund and Hashiguchi connections are of this type \cite{matsumoto86,antonelli93,anastasiei96,bejancu00,szilasi14,minguzzi14c}. .

Each Finsler connection $\nabla$ determines two covariant derivatives $\nabla^H$ and $\nabla^V$ respectively being obtained from $\nabla_{\check X}$ whenever $\check X$ is the horizontal or the vertical lift of a vector $X\in TM$. The horizontal covariant dervative $\nabla^H$ is determined by local connection coefficients $H^\alpha_{\mu \nu}(x,y)$, in particular over a vector field $Z\colon E\to VE\simeq TM$ it acts as follows
\[
(\nabla^H_\alpha Z)^\beta=\frac{\delta Z^\beta}{\delta x^\alpha}+H^\beta_{\mu \alpha} Z^\mu.
\]
The Berwald connection is determined by $V^\alpha_{\beta \gamma}=0$ and
\[
H^\alpha_{\mu \nu}:=G^\alpha_{\mu \nu}:=\frac{\p}{\p v^\nu}\,G^\alpha_\mu,
\]
Both the Chern-Rund and the Cartan connections are such that $\nabla^H g=0$, hence
\begin{equation} \label{sog}
H^\alpha_{\beta \gamma}:=\Gamma_{\beta \gamma}^{\alpha}:=\frac{1}{2} g^{\alpha \sigma} \Big( \frac{\delta}{\delta x^\beta} \,g_{\sigma \gamma}+\frac{\delta}{\delta x^\gamma} \, g_{\sigma \beta}-\frac{\delta}{\delta x^\sigma}\, g_{\beta \gamma}\Big),
\end{equation}
however, for the former $V^\alpha_{\beta \gamma}=0$ while for the latter $V^\alpha_{\beta \gamma}=C^\alpha_{\beta \gamma}$. We shall denote with $\nabla^{HC}$ the horizontal covariant derivative for the Cartan (Chern-Rund) connection, and with $\nabla^{HB}$ the horizontal covariant derivative for the Berwald connection.  The difference $G_{\beta \gamma}^{\alpha}-\Gamma_{\beta \gamma}^{\alpha}=L_{\beta \gamma}^{\alpha}$ is the Landsberg tensor \cite{shen01,minguzzi14c}.

\section{The divergence theorem}

Let $s\colon M \to TM\backslash 0$ be a section, and let us consider the pullback
metric $s^*g$. Its components are $(s^*g)_{\alpha \beta}(x)=g_{\alpha
\beta}(x,s(x))$. Let $\nabla^{s^*g}$ be the Levi-Civita connection of $s^*g$, and let
$\overset{s}{\nabla}$ be the pullback of a Finsler connection. Its coefficients are
\cite[Sec.\ 4.1.1]{minguzzi14c} \cite[Eq.\ (3.7)]{ingarden93}
\[
(\overset{s}{\nabla})^\gamma_{\alpha \beta}=H^\gamma_{\alpha \beta}(x,s(x))+V^\gamma_{\beta \mu}(x,s(x))D_\alpha s^\mu,
\]
where $H^\gamma_{\alpha \beta}$ and $V^\gamma_{\alpha \beta}$ are the coefficients of
the horizontal and vertical covariant derivatives of the Finsler connection.
Particularly simple are the coefficients of the pullback
$\overset{s}{\nabla}{}^{\mathrm{ChR}}$ of the Chern-Rund connection, namely
$\Gamma^\gamma_{\alpha \beta}(x,s(x))$. Since the connection coefficients on $M$
transform with the usual law, their difference is a tensor which we wish to
calculate.

\begin{proposition}
For every $X,Y\colon M \to TM$
\begin{equation} \label{dos}
\nabla^{s^* g}_X Y-\overset{s}{\nabla}{}^{\mathrm{ChR}}_X Y=C(D_Ys,X)+C(D_Xs,Y)-g(C(X,Y),D
s)^\sharp.
\end{equation}
\end{proposition}

It is understood that in this formula $C$ stands for $s^* C$ and $\sharp: T^*M\to TM$ is the musical isomorphism given by the metric $s^*g$.

\begin{proof}
With obvious meaning of the notation, and lowering the upper index of the connection
to the left, we have at $(x,s(x))$
\begin{align*}
2 g_{\alpha \delta} &\big[(\nabla^{s^*g})^\alpha_{\beta \gamma}\!-\!\Gamma^\alpha_{ \beta
\gamma}(x,s(x))\big] =\frac{\p}{\p x^\beta } (s^*g)_{\delta \gamma}\!+\!\frac{\p}{\p
x^\gamma} (s^*g)_{\delta \beta}  \!-\!\frac{\p}{\p x^\delta } (s^*g)_{\beta \gamma}\!-\!2
\Gamma_{\delta \beta \gamma}\nonumber\\
&=\big(\frac{\p}{\p x^\beta } g_{\delta \gamma}+\frac{\p}{\p x^\gamma} g_{\delta
\beta}-\frac{\p}{\p x^\delta } g_{\beta \gamma}\big) (x,s(x)) -2 \Gamma_{\delta \beta
\gamma} + 2 C_{\mu \delta \gamma} \frac{\p s^\mu}{\p x^\beta}\\&\qquad+2 C_{\mu \delta \beta}
\frac{\p s^\mu}{\p x^\gamma}-2 C_{\mu \beta \gamma} \frac{\p s^\mu}{\p
x^\delta}\nonumber\\
&=\big(\frac{\delta}{\delta x^\beta } g_{\delta \gamma}+\frac{\delta}{\delta
x^\gamma} g_{\delta \beta}-\frac{\delta}{\delta x^\delta } g_{\beta \gamma}\big)
(x,s(x)) -2 \Gamma_{\delta \beta \gamma} \\&\qquad+ 2 C_{\mu \delta \gamma} D_\beta s^\mu +2
C_{\mu \delta \beta} D_\gamma s^\mu-2 C_{\mu \beta \gamma} D_\delta s^\mu
\nonumber\\
&=2 C_{\mu \delta \gamma} D_\beta s^\mu +2 C_{\mu \delta \beta} D_\gamma s^\mu-2
C_{\mu \beta \gamma} D_\delta s^\mu . \nonumber
\end{align*}
Thus
\begin{equation*}
(\nabla^{s^*g})^\alpha_{\beta \gamma}-\Gamma^\alpha_{ \beta \gamma}(x,s(x))=
C^\alpha_{\mu  \gamma} D_\beta s^\mu + C^\alpha_{\mu \beta} D_\gamma s^\mu- C_{\mu
\beta \gamma}  g^{\alpha \delta} D_\delta s^\mu
\end{equation*}
\end{proof}
Equation (\ref{dos}) clarifies that the connection of the pullback metric is neither
the pullback of the Chern-Rund connection, nor the pullback of the Cartan connection
(known as Barthel connection \cite[Theor.\ 1]{ingarden93}).

\begin{corollary}
For every vector field $Y\colon M\to TM$,
\[
\nabla^{s^* g} \cdot
Y=\overset{s}{\nabla}{}^{\mathrm{ChR}} \cdot Y+s^*I(D_Y s).
\]
\end{corollary}

We need a lemma to remove the derivative of the pullback connection in favor of the
derivative of the horizontal connection.
\begin{lemma}
Let $Z\colon TM\backslash 0\to TM$ and let $s\colon M \to TM\backslash 0$, then
\begin{equation}
\overset{s}{\nabla}{}^{\mathrm{ChR}}_X s^*Z=s^*(\nabla^{HC}_X Z)+ s^*\Big(\frac{\p Z^\mu}{\p y^
\beta}\Big) D_X s^\beta \frac{\p}{\p x^\mu}
\end{equation}
\end{lemma}

\begin{proof}
It follows from
\begin{align*}
\frac{\p}{\p x^\gamma} Z^\mu(x,s(x))&=\frac{\p Z^\mu}{\p x^\gamma}+\frac{\p
Z^\mu}{\p y^\beta } \frac{\p s^\beta}{\p x^\gamma}=\Big(\frac{\p Z^\mu}{\p
x^\gamma}-N^\beta_\gamma \frac{\p Z^\mu}{
\p y^\beta}\Big)+ \frac{\p Z^\mu}{\p y^\beta } \Big(\frac{\p s^\beta}{\p
x^\gamma}+N^\beta_\gamma\Big)\\
&=s^*\Big(\frac{\delta Z^\mu}{\delta x^\gamma}\Big)+s^*\Big(\frac{\p Z^\mu}{\p y^\beta}\Big) D_\gamma
s^\beta,
\end{align*}
adding $\Gamma^\mu_{\sigma \gamma}(x,s(x)) Z^\sigma(x,s(x))$ to both sides.
\end{proof}

As a consequence we obtain the following expression for the divergence, see also
\cite[Eq.\ 3.10]{rund75}.

\begin{theorem} \label{yqq}
Let $Z\colon TM\backslash 0\to TM$ and let $s\colon M \to TM\backslash 0$, then
\begin{equation} \label{oap}
\nabla^{s^* g} \cdot s^*Z=s^*(\nabla^{HC} \cdot Z)+s^*I(D_{s^*Z}s)+ s^*\Big(\frac{\p
Z^\mu}{\p y^
\beta}\Big) D_\mu s^\beta
\end{equation}
\end{theorem}
Integrating we recover Rund's  divergence theorem \cite[Eq.\ (3.17)]{rund75}
\begin{equation} \label{hyd}
\int_D\! \left(\!s^*(\nabla^{HC}\!\! \cdot\! Z)+s^*I(D_{s^*Z}s)+ s^*\Big(\frac{\p
Z^\mu}{\p y^
\beta}\Big) D_\mu s^\beta \!\right)\! \mu(s)=\!\int_{\p D}\!\!\!\!\!\! s^*g(s^*Z, n^R(s))  \nu^R(s)
\end{equation}
where $\mu(s)=\sqrt{\vert \det s^*g\vert}\,\dd x$ is the canonical volume of the Riemannian metric $s^*g$,  $\nu^R(s)$ is the canonical volume on the hypersurface $\p D$ associated to the induced metric $s^*g\vert_{\p D}$, and $n^R(s)$ is the outward normal to $\p D$ according to $s^*g$ (hence $T_p \p D=\ker g_{s(p)}(n^R(s),\cdot)$). The letter $R$ stands for ``Rund's''.


The problem with this formulation is, of course, that in this way one would get a boundary term  in which the normal is determined  according to $s^*g$ and the area form according to the metric induced by $s^*g$ on the boundary. It is really a (pseudo-)Riemannian divergence theorem in disguise rather than a Finslerian divergence theorem.

We are now going to use Eq.\ (\ref{oap}) in order to get a genuinely Finslerian divergence theorem
in integral form. By this we mean that the pullback metric and its Levi-Civita
connection will not appear, and also the area form  and normal at the boundary will
be  Finslerian.

Let $\mu$ be a volume form on $M$ and let  $S\subset M$ be a hypersurface with normal
$n$, meaning by this that at every $p\in S$, $T_pS=\ker g_n(n,\cdot)$. Here we are using the fact that whenever the manifold dimension is at least 3 the Legendre map $\ell \colon v \mapsto g_v(v,\cdot)$ is a diffeomorphism of $T_pM\backslash 0$. The author provided a proof  in \cite[Theor.\ 6]{minguzzi13c} and an independent and different one can be found in  Ruzhansky and Sugimoto \cite{ruzhansky15}. In 2 dimensions the signature of $g$ can only be Riemannian or Lorentzian. In the former case it is well known that $\ell$ is a diffeomorphism, in the latter case it is not difficult to show that $\ell$ is surjective but not necessarily injective. So in more than two dimensions the normal exists and is unique up to scaling, while in two dimensions the normal exists but is non-unique, as there can be many with  different directions. Observe that  the section $s$ does not enter the definition of $n$, so this normal is different from Rund's $n^R(s)$. Let $X$ be a
field transverse to $S$, then the volume form induced on $S$ is
\begin{equation} \label{ind}
\nu=\frac{1}{-g_{n}(n,X)} \, i_X \mu,
\end{equation}
and is evaluated only on vectors tangent to $S$. It is easy to show that it is
independent of the choice of transverse field $X$ and depends only on the scale of
$n$. If the signature of $g$ is Lorentzian, namely $(-,+,\cdots,+)$, and $S$ is spacelike the scale can be fixed with $g_n(n,n)=-1$, in which case $\nu_S=i_n
\mu$.

Let $D$ be a domain such that $\p D=S$. By the divergence theorem
\[
\int_D \textrm{div}_\mu X \mu=\int_D d\, i_X\mu=\int_S i_X \mu=-\int_S g_n(n,X) \nu.
\]
Now, if the mean Cartan torsion \[I_\alpha:=g^{\mu \nu} C_{\mu \nu \alpha}=\frac{1}{2} \frac{\p}{\p y^\alpha} \log \vert\det g_{\mu \nu}\vert\] vanishes, as the determinant does not depend on the fiber variable, we can define on $M$ the usual natural volume
form $\mu=\sqrt{\vert \det g\vert}\,\dd x=\sqrt{\vert \det s^*g\vert}\,\dd x$, thus using Theor.\ \ref{yqq} we arrive at the next result.

\begin{theorem} \label{diu}
Let $(M,\mathscr{L})$ be a pseudo-Finsler space for which the mean Cartan torsion
vanishes. Let $Z\colon TM\backslash 0\to TM$ and let $s\colon M \to TM\backslash 0$,
then
\begin{equation}
\int_D\Big\{s^*(\nabla^{HC} \cdot Z)+ s^*\Big(\frac{\p Z^\mu}{\p y^
\beta}\Big) D_\mu s^\beta\Big\} \, \mu= -\int_{\p D} g_n(n,s^* Z) \nu .
\end{equation}
\end{theorem}
The normal $n$ and the  form $\nu$ do not depend on the section $s$.

\begin{remark}
This theorem does not follow from Rund's divergence theorem \cite[Eq.\ (3.17)]{rund75} (see Eq.\ (\ref{hyd})) by setting $I_\alpha=0$ since the normal and induced volume form have been obtained following a rather different geometrical argument. Observe that in  Rund's theorem   the boundary term   would still depend on the section, so it is interesting that it can be replaced by the more transparent form  given by Theorem \ref{diu}.
\end{remark}

\begin{remark}
Apparently the theorem privileges the Chern-Rund or the Cartan connections but it is not really so. In fact, in the divergence $\nabla^{HC} \cdot Z$ the connection coefficients enter only in the combination \cite[Eq.\ (59)]{minguzzi14c} $\Gamma^\mu_{\alpha \mu}=\frac{\delta}{\delta x^\alpha} \ln \sqrt{\vert \det g\vert}$. Recall that $G^\mu_{\alpha \mu}=\Gamma^\mu_{\alpha \mu}+L^\mu_{\alpha \mu}$. We have shown \cite[Eq.\ (52)]{minguzzi14c})  that whenever $I_\alpha=0$ we have $J_\alpha:=L^\mu_{\alpha \mu}=0$, thus  $G^\mu_{\alpha \mu}=\Gamma^\mu_{\alpha \mu}$ and we can replace $\nabla^{HC} \cdot Z$ with $\nabla^{HB} \cdot Z$.
\end{remark}

\subsection{Symmetry implies conservation}

Let us suppose that $Z$ is a vertical gradient, $Z_\gamma=\frac{\p}{\p y^\gamma} f$, $f\colon E\to \mathbb{R}$, and divergenceless in a horizontal Finslerian sense, namely $\nabla^{HC} \cdot Z=0$.
We want to show that any pregeodesic Killing vector field $s$ implies a conserved quantity.

If $s$ is a Finslerian Killing vector \cite{knebelman29,rund59}
\begin{equation} \label{kil}
g_{ \delta  \beta} \nabla^{HC}_\gamma s^\beta+ g_{\gamma \beta}\nabla^{HC}_\delta s^\beta+
2 y^\beta (\nabla^{HC}_\beta s^\mu) C_{\gamma \delta \mu}  =0.
\end{equation}
which evaluated at the support vector $s$, and setting $D^\beta=g^{\beta\alpha}(x,s(x)) D_\alpha$, gives
\begin{align}
 & D^\gamma s^\delta+ D^\delta s^\gamma+ 2 (D_s
 s^\mu) C^{\gamma \delta}_{\mu}  =0,
\end{align}
We can write
\begin{equation}
\begin{aligned}  \label{hud}
s^*\Big(\frac{\p Z^\mu}{\p y^
\delta}\Big) D_\mu s^\delta \! &=s^*\Big(g^{\mu \nu} \frac{\p Z_\nu}{\p y^\delta}-\!2 C^{\mu \nu}_\delta Z_\nu\Big)D_\mu s^\delta=s^*\Big( \frac{\p Z_\gamma}{\p y^\delta}\!-\!2 C^{\nu}_{\gamma\delta} Z_\nu\Big) D^\gamma s^\delta \\
&=-s^*\Big( \frac{\p Z_\gamma}{\p y^\delta}-2 C^{\nu}_{\gamma\delta} Z_\nu\Big)  C^{\gamma \delta}_{\mu} (D_s
 s^\mu)
 \end{aligned}
\end{equation}
If $s$ is  pregeodesic, $D_s s\propto s$, this term vanishes and so, if $(M,\mathscr{L})$ is a globally hyperbolic spacetime, by Theor.\ \ref{diu} the
quantity
\begin{equation}
E=\int_S  g_n(n, s^*Z) \nu ,
\end{equation}
 is conserved, namely independent of the Cauchy hypersurface $S$, where $n$ is the
 normal to $S$ and $\nu$ is the induced  volume.
As a typical application, $Z$ will be the  energy-momentum current, $s$ will represent an observer as a future-directed timelike field, and $E$ will be the total energy content of spacetime according to observer $s$.

If $s$ is normalized, $g_s(s,s)=-1$, the geodesic assumption is superfluous, indeed from Eq.\ (\ref{kil}) we obtain  $2g_s(D_ss,
  \cdot)=-\dd (g_s(s,s))=0$, thus $s$ is geodesic.

The problem of energy-momentum conservation in Finsler gravity theories is notoriously difficult \cite{rund62,ishikawa80,ikeda81,anastasiei87,rutz96}. Here we have  given  sufficient conditions on $Z$ that imply the existence of a conserved energy.

The divergence of $Z$  can be elaborated as follows, see \cite[Eq.\ (16)]{minguzzi14c}
\begin{align}
\nabla^{HC} \!\cdot\! Z&=g^{\mu \nu} \nabla^{HC}_\mu Z_\nu =g^{\mu \nu} \Big(\frac{\delta}{\delta x^\mu} \frac{\p f}{\p y^\nu} -\Gamma^\alpha_{\nu \mu} \frac{\p f}{\p y^\alpha}\Big)=g^{\mu \nu} \Big(\frac{\p }{\p y^\nu} \frac{\delta f}{\delta x^\mu}  +L^\alpha_{\nu \mu} \frac{\p f}{\p y^\alpha}\Big)\nonumber \\
&=\frac{\p }{\p y^\nu} \Big(g^{ \nu \mu}  \frac{\delta f}{\delta x^\mu}\Big) +2I^\mu \frac{\delta f}{\delta x^\mu}+J^\alpha \frac{\p f}{\p y^\alpha} .
\end{align}
Since we assume $I_\alpha=0$ which implies $J_\alpha=0$ by \cite[Eq.\ (52)]{minguzzi14c}), the divergenceless condition can be equivalently written $\frac{\p }{\p y^\nu} \Big(g^{ \nu \mu}  \frac{\delta f}{\delta x^\mu}\Big)=0$. In conclusion, the conditions on the field $Z$ can be equivalently expressed through one of the following conditions on $f$: (a) the vertical divergence of the horizontal gradient of $f$ vanishes;  (b) the horizontal divergence of the vertical gradient of $f$ vanishes.

%

\section{Conclusions}
We have discussed the advantages and drawbacks of the known Finslerian divergence theorems. For pseudo-Finsler spaces having a vanishing mean Cartan torsion we have provided a divergence theorem which might be regarded as genuinely Finslerian as it deals with a vector field that might depend on the fiber coordinates, and  a boundary integral which involves the Finslerian normal and the Finslerian induced volume form.

 In the last section we have  shown how to use this theorem in order to construct conserved quantities in pseudo-Finsler space. The conditions placed on the energy-momentum current $Z$ might help to select the correct dynamical field equations for Finsler gravity theory.

\section*{Acknowledgments}
A preprint version of this work  appeared in {arXiv}:1508.06053.


\end{document}